\documentclass[conference]{IEEEtran}
\IEEEoverridecommandlockouts
\usepackage{cite}
\usepackage{amsmath,amssymb,amsfonts}

\usepackage{algorithmic}
\usepackage{graphicx, subfigure}
\usepackage{textcomp}
\usepackage{xcolor}
\usepackage{amsthm}
\usepackage{verbatim}
\usepackage{multirow}
\usepackage{tabularx, booktabs}
\usepackage{array}
\newcolumntype{P}[1]{>{\centering\arraybackslash}p{#1}}
\newtheorem{theorem}{Theorem}
\newcommand\numberthis{\addtocounter{equation}{1}\tag{\theequation}}

\begin{document}

\title{Bayesian AirComp with Sign-Alignment Precoding 
for Wireless Federated Learning
%\thanks{This research was supported by the MSIT(Ministry of Science and ICT), Korea, under the ICT Creative Consilience program(IITP-2020-2011-1-00783) supervised by the IITP(Institute for Information \& communications Technology Planning \& Evaluation)}
}

\author{Chanho Park, Seunghoon Lee, and Namyoon Lee% <-this % stops a space
\thanks{ C. Park, S. Lee, and N. Lee are with the Department of Electrical Engineering, POSTECH, Pohang, Gyeongbuk, 37673 Korea (e-mail: \{chanho26, shlee14, nylee\}@postech.ac.kr).

This work was supported by Institute of Information \& communications Technology Planning \& Evaluation(IITP) grant funded by the Korea government(MSIT). (No.2021-0-00467, Intelligent 6G Wireless Access System)}} 
%\thanks{Y. C. Eldar is with the Math and CS Faculty, Weizmann Institute of Science,
%Rehovot, Israel (e-mail: yonina.eldar@weizmann.ac.il).}}

\begin{comment}
\author{\IEEEauthorblockN{Chanho Park}
\IEEEauthorblockA{\textit{Department of Electrical Engineering} \\
\textit{POSTECH}\\
Pohang, South Korea \\
chanho26@postech.ac.kr}
\and
\IEEEauthorblockN{Geon Choi}
\IEEEauthorblockA{\textit{Department of Electrical Engineering} \\
\textit{POSTECH}\\
Pohang, South Korea \\
simon03062@postech.ac.kr}
\and
\IEEEauthorblockN{Seunghoon Lee}
\IEEEauthorblockA{\textit{Department of Electrical Engineering} \\
\textit{POSTECH}\\
Pohang, South Korea \\
shlee14@postech.ac.kr}
\and
\IEEEauthorblockN{Namyoon Lee}
\IEEEauthorblockA{\textit{Department of Electrical Engineering} \\
\textit{POSTECH}\\
Pohang, South Korea \\
nylee@postech.ac.kr}
%\thanks{ C. Park, S. Lee, and N. Lee are with the Department of Electrical Engineering, POSTECH, Pohang, Gyeongbuk, 37673 Korea (e-mail: \{chanho26, shlee14, nylee\}@postech.ac.kr).} 
}
\end{comment}

\maketitle

\begin{abstract} 
In this paper, we consider the problem of wireless federated learning based on sign stochastic gradient descent (signSGD) algorithm via a multiple access channel. When sending locally computed gradient's sign information, each mobile device requires to apply precoding to circumvent wireless fading effects. In practice, however, acquiring perfect knowledge of channel state information (CSI) at all mobile devices is infeasible. In this paper, we present a simple yet effective precoding method with limited channel knowledge, called sign-alignment precoding. The idea of sign-alignment precoding is to protect sign-flipping errors from wireless fadings. Under the Gaussian prior assumption on the local gradients, we also derive the mean squared error (MSE)-optimal aggregation function called Bayesian over-the-air computation (BayAirComp).  Our key finding is that one-bit precoding with BayAirComp aggregation can provide a better learning performance than the existing precoding method even using perfect CSI with AirComp aggregation.

%Comm
%FL is a type of machine learning that enables learning for heterogeneous datasets present on each device without sharing private data directly. Considering the efficiency of learning, we propose a soft-signSGD based FL system utilizing AirComp technique. Unlike the existing precoding scheme of FL systems using AirComp, we reveal that it is possible to learn a training model with sufficient accuracy even if precoding with only 1-bit information about channels. Moreover, based on the SBFL algorithm, which enables accurate learning for heterogeneous data, we develop a new aggregation function suitable for AirComp from an MMSE perspective. Through simulation results that learn image classification models for MNIST and CIFAR10 datasets, we show that one-bit precoding can train models with sufficient accuracy and also has better performance in accuracy for heterogeneous datasets than a conventional system.
\end{abstract}

%\begin{IEEEkeywords}
%federated learning, over-the-air computation, quantization, one-bit precoding, prior distribution.
%\end{IEEEkeywords}

\section{Introduction}
Federated learning (FL) is a class of distributed machine learning technique using locally generated heterogenous datasets at mobile devices. Communicating between mobile devices and a central server, it can train a model accurately, while maintaining the privacy of data present in mobile devices \cite{mcmahan2017communication, konevcny2016federated}. Federated averaging (FedAvg) and federated stochastic gradient descent (FedSGD) are the representative algorithms for FL. In FedSGD, mobile devices send locally computed gradient information to the server, and the server aggregates the local gradients to update the global model parameters. To improve learning efficiency for FL, the variations of FedAvg and FedSGD have been proposed in \cite{alistarh2017qsgd, reisizadeh2020fedpaq, chen2016revisiting, li2019convergence}. 

%In FedAvg, mobile devices update model parameters using their local datasets and send them to the server. By taking a weighted average of the received local model parameters, the server updates the global model parameters. 

%Communication-efficient algorithms for wireless FL are indispensable to minimize extremely large communication costs. 

Over-the-air computation (AirComp)-based FL has been recently proposed as a communication-bandwidth efficient aggregation method \cite{zhu2019broadband, zhu2020one, amiri2020machine, amiri2020federated, seif2020wireless, yang2020federated}. Using the superposition property of wireless medium,  AirComp performs wireless analog aggregation of the local gradients on the fly. This approach can attain low-latency learning performance compared with the case of using orthogonal access techniques when implementing FedSGD in a wireless setup. In addition,  AirComp enhances the security of individual data because it makes difficult to estimate individual local gradient information.  In AirComp, precoding for aligning the local gradients is essential to circumvent heterogenous channel fading effects across mobile devices \cite{zhu2019broadband, zhu2020one, amiri2020machine, amiri2020federated, seif2020wireless, yang2020federated}. Several precoding strategies have been presented, including truncated-channel inversion precoding \cite{zhu2019broadband} and dithering-based precoding \cite{shlezinger2020federated}. The underlying idea of the precoding strategies is to perform pre-equalization to mitigate fading effects; thereby, the server can receive a superposition of aligned local gradients on the fly. These alignment precoding methods, however, can be challenging to implement in wireless FL systems. Mobile devices located far from the server may be infeasible to consistently apply channel-inversion-based precoding, while satisfying the power constraint. Besides, it is challenging to acquire perfect channel state information (CSI) for uplink communications in practice. 
% Several precoding strategies have recently been proposed, which includes a truncation-based strategy [] and unbiased gradient aggregation, [4], a joint design of the transmit powers and a denoising factor to minimize the mean-squared-error (MSE) between the estimated gradient and the true gradient at each iteration. The aforementioned methods, however, are limited to use for the quantized SGD algorithm 

In this paper, we consider a sign stochastic gradient descent (signSGD) algorithm \cite{bernstein2018signsgd} for FL over a shared wireless multiple access channel. Sign-SGD is a communication-efficient distributed learning algorithm. This algorithm can reduce the communication cost because it exploits the sign information of local gradients when updating the model. Besides, it can be implementable using simple binary digital modulated transmission techniques in wireless FL settings \cite{zhu2020one}. Each mobile device performs one-bit quantization of the locally computed gradient in every communication round to minimize uplink communication cost. Then, it transmits the sign of local gradient along with precoding to mitigate channel fadings using a shared time-frequency resource. Then, the server receives a superposition of precoded local gradient signs. Using this received gradient information, the server updates the model parameters and shares them with the mobile devices for the next round iteration.

Our main contribution is to propose novel precoding called sign-alignment precoding. The idea of our precoding is to align the sign of the channel fading coefficient to avoid gradients' sign flipping errors by fadings. This precoding requires one-bit CSI at transmitter (CSIT) information; thereby, it can significantly reduce the channel acquisition and feedback overheads for wireless FL compared to the conventional FL algorithm using channel-inversion based precoding, which requires full CSIT at mobile devices. We also present a novel Bayesian aggregation method for AirComp, referred to as BayAirComp. Inspired by our prior work in \cite{lee2020bayesian}, the key idea of BayAirComp is to map the received signal to the estimate of the sum of local gradients to minimize the mean squared error (MSE) by harnessing the knowledge of prior distributions of local gradients as side-information. We present experimental results to show that sign-alignment precoding with BayAirComp can outperform the state-of-the-art one-bit broadband digital aggregation (OBDA) algorithm \cite{zhu2020one}.

\section{System Model}
In this section, we describe learning and communication models for a wireless FL system.  The wireless FL system consists of $K$ mobile devices and a server (or a base station) as depicted in Fig. \ref{Fig1}. The server trains a neural network with a large number of hyper-parameters ${\bf w}\in \mathbb{R}^{M}$ by communicating with $K$ mobile devices through a shared wireless channel. 

%We consider a wireless FL system, in which a server collaboratively trains a neural network model by communicating with $K$ mobile devices, each with d 

% the aggregation of local gradients transmitted by mobile users is performed through AirComp. The subsections below describe this system by dividing into a learning model and a communication model.

\subsection{Loss function}
Let $\mathbf{z}_k^i \in \mathbb{R}^{d}$ and $r_{k}^{i} \in \mathbb{R}$ be the $i$th pair of the training data example stored at mobile device $k\in[K]$. Assuming, device $k$ has $N_{k}$ training examples, we define a set of training examples stored at device $k\in [K]$ as $\mathcal{D}_{k}=\{\mathbf{z}_{k}^{i}, r_{k}^{i}\}_{i=1}^{N_{k}}$. Therefore, a total number of training examples for learning becomes $N=\sum_{k=1}^{K} N_{k}$.  Given model parameter ${\bf w}\in \mathbb{R}^{M}$, we define a loss function with training pair $(\mathbf{z}_k^i, r_k^i)$ as $\ell\left(\mathbf{z}_k^i, r_k^i; {\bf w}\right):\mathbb{R}^M\times \mathbb{R}\rightarrow \mathbb{R}$. This loss can be either a cross entropy or a mean-squared error function according to machine learning applications. Using the sample average, the local loss function of device $k$ is defined as
\begin{align} \label{localloss}
    f_{k}\left( \mathbf{w} \right) = \frac{1}{N_{k}} \sum_{i=1}^{N_{k}} \ell\left( \mathbf{z}_k^i, r_k^i; \mathbf{w} \right).
\end{align} 
Summing $f_{k}\left( \mathbf{w} \right)$ with weight $\frac{N_k}{N}$ for $k\in [K]$, the global loss function is given by
\begin{align}
    F\left( \mathbf{w} \right) = \frac{1}{N} \sum_{k=1}^{K} \sum_{i=1}^{N_k} \ell\left( \mathbf{z}_k^i, r_k^i; \mathbf{w} \right) = \sum_{k=1}^{K} \frac{N_k}{N} f_k\left( \mathbf{w} \right).
\end{align}

\subsection{Wireless federated learning via singSGD}
The wireless FL system iteratively optimizes model parameter ${\bf w}$ over $T$ communication rounds. Each communication round comprises 1) gradient computation and compression, 2) uplink transmissions, 3) model update, and 4) downlink transmission.

{\bf 1) Gradient computation and compression:} In communication round $t\in [T]$, mobile device $k\in[K]$ first computes local gradient information. Let $\mathbf{g}_k^t \triangleq \nabla f_k \left( \mathbf{w}^{t} \right)$ be the local gradient evaluated using model knowledge ${\bf w}^{t}$ and local data set $\mathcal{D}_k$ for $k\in [K]$. Then, it compresses local gradient using one-bit quantizer to diminish the uplink communication cost as
\begin{align}
	 \mathbf{\hat{g}}_k^t = {\sf sign} \left( \mathbf{g}_k^t \right),
\end{align}
where ${\sf sign}(x)=1$ for $x\geq 0$ and ${\sf sign}(x)=-1$ otherwise.

\begin{figure}[t]
    \centering 
    \includegraphics[width=1\columnwidth]{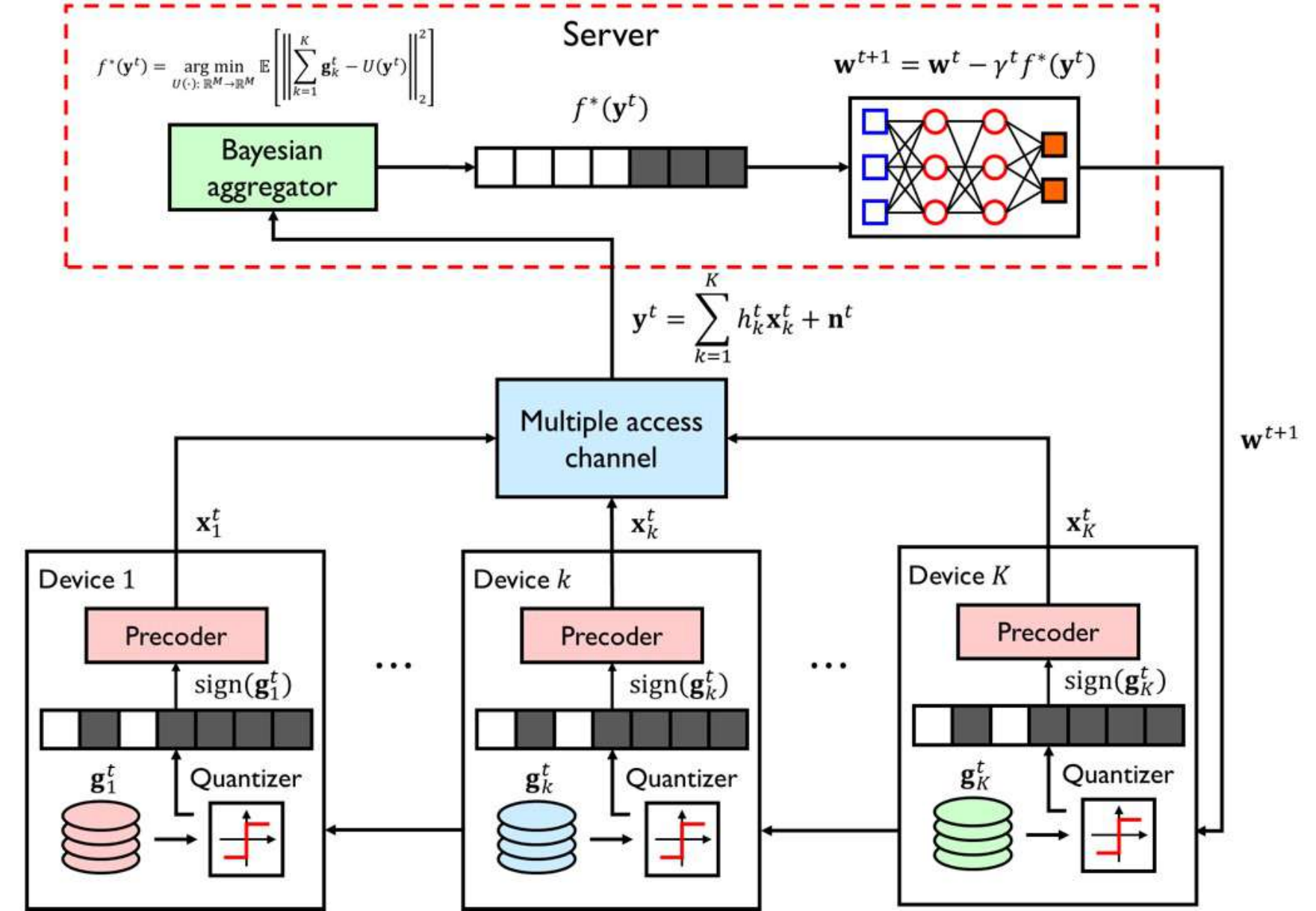}
    \vspace{-0.3cm}\caption{Illustration of the wireless signSGD framework, in which mobile devices jointly train a model over a shared multiple access channel.}\label{Fig1}
\end{figure}

{\bf 2) Uplink communications with precoding via a MAC:} After the compression, mobile device $k$ transmits binary vector $\mathbf{\hat{g}}_k^t$ along with the precoding coefficient $v_k^t\in \mathbb{R}$. 
\begin{align}
	{\bf x}_k^t=v_{k}^t{\sf sign} \left( \mathbf{g}_k^t \right),
\end{align}
where the precoder $v_k^t$ satisfies the average power constraint as 
\begin{align}
    \mathbb{E} \left[ \left| v_{k}^t \right|^2 \right] \le P.
\end{align}
This compressed and precoded gradient information is sent over a shared multiple access channel.  We focus on the real part of the complex-baseband signal model for ease of exposition. Let  $h_k^t\in \mathbb{R}$ be the real part of the complex baseband channel fading coefficient from mobile device $k$ to the server at communication round $t$.  We consider a block fading channel model, in which $h_k^t$ independently changes over different communication rounds, while it remains as a constant per communication round. Then, under the premise of perfect synchronization, the received signal is given by
\begin{align} \label{Aircomp}
    {\bf y}^t &= \sum_{k=1}^K h_{k}^t{\bf x}_k^t + {\bf n}^t, %\nonumber\\
   % & = \sum_{k=1}^K h_{k}^t v_{k}^t{\sf sign} \left( \mathbf{g}_k^t \right)+ {\bf z}^t,
\end{align}
where ${\bf n}^t$ is the real part of the complex-baseband noise signal at the server, which is distributed as independent and identically distributed (IID) Gaussian, i.e., ${\bf n}^t\sim \mathcal{N} \left(0,\frac{1}{2}{\bf I}_M\right)$.  

% and the additive noise $z_i^t$ follows i.i.d. Gaussian distribution with variance $\sigma_z^2$, i.e., $h_{k}^t \sim \mathcal{N} (0,1)$ and $z_i^t \sim \mathcal{N}(0, \sigma_z^2)$. $p_{k,i}^t$ is the precoder of $k$th device, and each device has the transmit power constraint as 
{\bf 3) Model update:} In communication round $t\in [T]$, the server performs the update of the model parameter using the gradient descent algorithm \cite{mcmahan2017communication}. To perform the gradient decent algorithm, the server requires to estimate the sum of local gradients from the received signal ${\bf y}^t$, which is a noisy version of the sum of faded local gradients. Let $U(\cdot):\mathbb{R}^M\rightarrow \mathbb{R}^M$ be the sum gradient estimator. Then, the MSE-optimal gradient estimator is defined as
\begin{align}
	 f^{\star}({\bf y}^t) =\arg\min_{ U(\cdot):\mathbb{R}^M\rightarrow\mathbb{R}^M}\mathbb{E}\left[ \left\|  \sum_{k=1}^K{\bf g}_k^t - U({\bf y}^t)\right\|_2^2\right].
\end{align}
Using this MSE-optimal estimator, the server updates the model parameter with learning rate $\gamma^t\in \mathbb{R}^{+}$ at communication round $t$ for the next round iteration:
\begin{align} \label{GD}
    \mathbf{w}^{t+1} = \mathbf{w}^{t} - \gamma^t  f^{\star}({\bf y}^t).
\end{align}
 
\begin{figure*}[h]
\begin{align} 
    f_{\sf{BayAirComp}} \left( y_m^t \right) &=
    \frac{1}{K} \sum_{k=1}^K \left[ \mu_k^t + 
    \sqrt{\frac{2}{\pi}} \nu_k^t 
    \frac{\sum_{\mathbf{b} \in \mathcal{B}_{K,k}} \exp \left[-\frac{\left( y_m^t - \left( \mathbf{h}^t \right)^T \mathbf{b} \right)^2}{2\sigma^2} \right] - \sum_{\mathbf{b} \in \left( \mathcal{B}_{K,k} \right)^{c}} \exp \left[-\frac{\left( y_m^t -(\mathbf{h}^t)^{\top} \mathbf{b} \right)^2}{2\sigma^2} \right] }{\sum_{\mathbf{b} \in \mathcal{B}_K} \exp \left[-\frac{\left( y_m^t - \left( \mathbf{h}^t \right)^T \mathbf{b} \right)^2}{2\sigma^2} \right] }
    \right]. 
    \numberthis{\label{MMSEaggregate}}
\end{align}
\hrulefill
\end{figure*} 

{\bf 4) Downlink communication:} Using the broadcast nature of the wireless medium, the server multicasts the updated model parameter $\mathbf{w}^{t+1}$ to mobile devices using a shared downlink channel.  We assume that all mobile devices can perfectly decode the updated model parameters over entire communication rounds for ease of exposition.

\section{BayAirComp with Sign-Alignment Precoding}
In this section, we present a novel wireless federated learning algorithm. The key idea of the proposed algorithm entails two operations: 1) sign-alignment precoding in the uplink transmission and 2) the Bayesian AirComp aggregation in the reception.

%We propose a signSGD-based FL system with AirComp in terms of precoding and aggregation to improve the efficiency of learning and model accuracy. In the wireless communication, precoding helps to estimate transmitted symbols accurately at the receiver with channel state information at the transmitter (CSIT). To obtain CSIT, we design frequency division duplex (FDD) system in which the receiver estimates channel coefficients using pilot signals, and feedback the channel estimate to the transmitter. In subsection B, we propose a novel precoding scheme which requires less channel information with limited feedback but also gives sufficient performance.
%
%The aggregation process entails lower training loss in sum-gradient estimation, thereby improving the model accuracy. There is a conventional FL scheme suitable for heterogeneous dataset, \textit{scalable Bayesian federated learning} (SBFL) \cite{lee2020bayesian}. First, \textit{Bayesian federated learning} (BFL) is a FL scheme which estimates sum-gradient in a minimum mean-square error (MMSE) perspective using prior distributions of local gradients, quantization function used for gradients, and channel information. SBFL is an efficient BFL algorithm adding an assumption that all components of local gradients are i.i.d. Gaussian distributed. We devised a new aggregation scheme suitable for AirComp based on SBFL, and it will be explained in subsection A and C.

\subsection{Local Gradient Parameter Estimation and Compression} \label{graddis}

%To accomplish the Bayesian aggregation, the server requires to know the prior distribution of local gradients ${\bf g}_{k,i}^t \in \mathbb{R}^M$. The prior distribution Characterizing the exact prior distribution of the local gradient is practically infeasible,  
 
To implement the Bayesian aggregation in \cite{lee2020bayesian}, the server requires to know the prior distribution of local gradients ${\bf g}_{k}^t \in \mathbb{R}^M$. Unfortunately, it is infeasible to characterize the exact prior distribution of ${\bf g}_{k}^t$ because it depends on both the local data distribution and deep neural network structures. Instead, we model the prior distribution of ${\bf g}_{k}^t$ as Gaussian with the moment matching technique \cite{lee2020bayesian}. Specifically, let $g_{k,m}^t$ be the $m$th entry of ${\bf g}_{k}^t$. We model that  $g_{k,m}^t$ follows IID Gaussian with mean $\mu_k^t$ and variance $\left( \nu_k^t \right)^2$, i.e., $g_{k,m}^t \sim \mathcal{N}(\mu_k^t, \left( \nu_k^t \right)^2 )$. The mean and variance are estimated by taking the sample average estimator as
\begin{align}
     \mu_k^t = \frac{1}{M}\sum_{m=1}^M g_{k,m}^t~~{\rm and}~~    \left( \nu_k^t \right)^2 = \frac{1}{M}\sum_{m=1}^M \left( g_{k,m}^t -\mu_k^t \right)^2.
\end{align}
  Although this Gaussian approximation on the prior distribution is not exact, it not only allows the Bayesian aggregation computationally tractable but also improves learning performance when training CNNs using MNIST datasets \cite{lee2020bayesian} in an orthogonalized multiple access channel environment. After computing the moments,  each device normalizes the local gradient by subtracting its mean:
 \begin{align}
    \bar{g}_{k,m}^t = g_{k,m}^t - \mu_k^t.
\end{align}
Afterwards, mobile device $k$ compresses the local gradient to diminish the uplink communication cost by using one-bit quantizer as
\begin{align}
	\mathbf{\hat{g}}_k^t = \text{sign} \left( \mathbf{\bar{g}}_k^t \right).
\end{align}
Mobile device $k\in [K]$ sends $\mathbf{\hat{g}}_k^t \in \{-1,+1\}^M$ with $\mu_k^t\in \mathbb{R}^{+}$ and $\nu_k^t\in \mathbb{R}^{+}$ to the server.

\subsection{Sign-alignment precoding}
When sending compressed local gradient, $\mathbf{\hat{g}}_k^t$, device $k$ requires to use precoder to compensate for the effect of wireless fading $h_k^t$. Unlike the prior approaches to invert the fading coefficient for precoding, we take a novel precoding strategy that requires one-bit CSI feedback from the BS. Our proposed precoding strategy is to \textit{align the signs} of local gradients using \textit{one-bit precoding}, i.e., $v_k^t={\sf sign}(h_k^t)$ as
\begin{align}
	{\bf x}_k = {\sf sign}(h_k^t){\bf \hat g}_k^t.
\end{align}
The received signal at the BS becomes
\begin{align} \label{Aircomp2}
  %  {\bf y}^t &= \sum_{k=1}^K h_{k}^t {\sf sign}(h_k^t){\bf \hat g}_k^t+ {\bf n}^t \nonumber\\
   {\bf y}^t  & = \sum_{k=1}^K |h_{k}^t|{\bf \hat g}_k^t+ {\bf n}^t.
\end{align}
This precoding strategy ensures to align the signs of local gradients. This sign alignment effect helps to estimate the sum-gradient accurately by avoiding the sign flipping errors due to the wireless channel fadings. In addition, this precoding strategy requires only one-bit CSI overhead compared to the conventional precoding system which needs $6\sim12$ bits for full CSI. Therefore, high efficiency can be obtained through this precoding strategy.
%This precoding strategy ensures to align the signs of local gradients. 
%By the rule of the proposed precoding scheme in \eqref{channelsign}, the server receives a signal $y_m^t$ as follows: 
%\begin{align} \label{channelsign}
%    p_{k,i}^t &= P_0 \hspace{0.1cm} \text{sign} \left( h_{k}^t \right), \\
%    \label{agg_sign}
%    y_m^t &= P_0 \sum_{k=1}^K \left| h_{k}^t \right| \hat{g}_{k,i}^t + z_i^t.
%\end{align}

\begin{figure*}[t]
    \centering 
    \includegraphics[width=2\columnwidth]{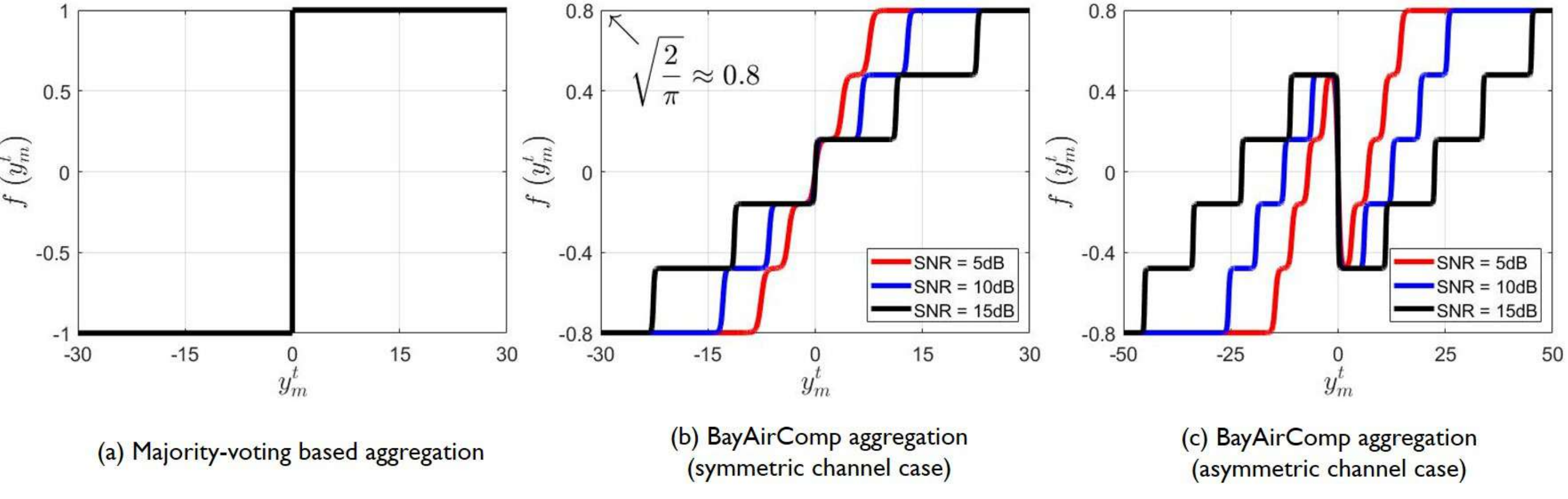}
    \vspace{-0.3cm}\caption{Comparison of aggregation functions according to fading coefficients and SNRs.}\label{Fig2}
\end{figure*}

%\begin{figure*}[h]
%    \centering
%    \subfigure[Majority-voting based aggregation function]
%    {
%        \includegraphics[width=0.63\columnwidth]{agg1.pdf}
%        \label{agg1}
%    }
%    \subfigure[BayAirComp aggregation function with symmetric channel]
%    {
%        \includegraphics[width=0.63\columnwidth]{agg2.pdf}
%        \label{agg2}
%    }
%    \subfigure[BayAirComp aggregation function with asymmetric channel]
%    {
%        \includegraphics[width=0.63\columnwidth]{agg3.pdf}
%        \label{agg3}
%    }
%    \caption{The majority-voting based aggregation function is presented in Fig. 2-(a). The proposed BayAirComp aggregation function in \eqref{MMSEaggregate} is presented in Fig. 2-(b) and (c). Fig. 2-(b) and 2-(c) show the function when $|h_k^t|=1$ for $k\in [5]$ and $|h_1^t|=5$ while $|h_k^t|=1$ for $k\in \{2,\ldots, 5\}$, respectively.}
%    \label{graph_agg}
%\end{figure*}

\subsection{Bayesian AirComp}
Using received signal ${\bf y}^t$, the BS requires estimating the sum of local gradients $\sum_{k=1}^K{\bf g}_k^t$ to accomplish the model update via a stochastic gradient descent algorithm. We present a novel aggregation method called BayAirComp.  The key idea of BayAirComp is to estimate the sum of local gradients $\sum_{k=1}^K{\bf g}_k^t$ by jointly exploiting the knowledge of the prior distribution of ${\bf g}_{k}^t$ (i.e., $\mu_k^t$ and $\nu_k^t$), one-bit quantizer, and fading coefficient $h_k^t$ to minimize the MSE. The aggregation function for BayAirComp $f_{\sf{BayAirComp}}(y_m^t):\mathbb{R}\rightarrow\mathbb{R}$ aims to minimize MSE, i.e., 
\begin{align}
	f_{\sf{BayAirComp}}(y_m^t)=\arg\min_{f:\mathbb{R}\rightarrow\mathbb{R}}\mathbb{E} \left[ \left| f(y_m^t) - \sum_{k=1}^K g_{k,m}^t \right|^2 \right].
\end{align}
The following theorem suggests the aggregation function for BayAirComp in closed form.

\begin{theorem}
	Let $\mathbf{h}^t = \left[ |h_{1}^t|, |h_{2}^t|, \cdots, |h_{K}^t| \right]\in \mathbb{R}_{+}^K$ be the sign-aligned channel vector at communication round $t$. We also denote $\mathcal{B}_K=\{-1,1\}^K$ be a set with $2^K$ elements of binary vectors with length $K$, i.e., ${\bf b}\in \mathcal{B}_K$. We also define a subset $\mathcal{B}_{K,k}\subset \mathcal{B}_{K}$ that contains binary vectors whose $k$th component is fixed to one. Then, the BayAirComp aggregation function is given by \eqref{MMSEaggregate}.

\end{theorem}

%\begin{proof}
%	See Appendix A. 
%\end{proof}
 
\begin{proof}  From the Bayesian principle, the MSE-optimal aggregation function is obtained by computing the conditional expectation as
\begin{align*}
    \arg \underset{f:\mathbb{R}\rightarrow\mathbb{R}}{\min} \mathbb{E} \left[ \left| f \left( y_m^t \right) - \sum_{k=1}^K \bar{g}_{k,m}^t \right|^2 \right] 
    &= \mathbb{E} \left[ \left. \sum_{k=1}^K \bar{g}_{k,m}^t \right| y_m^t \right] \\
    &= \sum_{k=1}^K \mathbb{E} \left[ \bar{g}_{k,m}^t \middle \vert y_m^t \right], \numberthis{\label{MMSE estimator}}
\end{align*} 
where the last equality is the linearity of the expectation. The conditional expectation in \eqref{MMSE estimator} is computed as
\begin{align} \label{MMSE element}
    \mathbb{E} \left[ \left. \bar{g}_{k,m}^t \right| y_m^t\right] = \frac{\int_{-\infty}^{\infty} \bar{g}_{k,m}^t P \left( \left. y_m^t \right| \bar{g}_{k,m}^t \right) P \left( \bar{g}_{k,m}^t \right)\, d\bar{g}_{k,m}^t}{\int_{-\infty}^{\infty} P \left( \left. y_m^t \right| \bar{g}_{k,m}^t \right) P \left( \bar{g}_{k,m}^t \right)\, d\bar{g}_{k,m}^t}.
\end{align}
We define ${\bf \tilde g}_{m}^t=[\bar{g}_{1,m}^t, \bar{g}_{2,m}^t, \cdots, \bar{g}_{K,m}^t]$. Applying the one-bit precoding for the sign alignment, the channel likelihood distribution is 
\begin{align}
    P \! \left( \left. y_m^t \right| { g}_{m,1}^t,\ldots,  { g}_{m,K}^t \right) =  \frac{1}{\sqrt{2\pi\sigma^2}} e^{  -\frac{\left( y_m^t -  \sum_{k=1}^K |h_{k}^t| {\sf sign}\left(\bar{g}_{k,m}^t\right)  \right)^2}{2\sigma^2}}. \label{eq:likelihood}
\end{align}
\noindent To obtain $P \! \left( \left. y_m^t \right| {\bar g}_{m,k}^t \right)$, we need to marginalize \eqref{eq:likelihood} with respect to ${\bar g}_{m,1}^t, \ldots, {\bar g}_{m,k-1}^t,{\bar g}_{m,k+1}^t,\ldots, {\bar g}_{m,K}^t $, where ${\bar g}_{k,m}^t \sim \mathcal{N} \left( 0, \left( \nu_k^t \right)^2 \right)$. Then, the marginal distribution $P \! \left( \left. y_m^t \right| { g}_{m,k}^t \right)$, we compute the numerator in \eqref{MMSE element} as 

%To obtain $P \! \left( \left. y_m^t \right| {\bar g}_{m,k}^t \right)$ from $P \! \left( \left. y_m^t \right| {\bar g}_{m,1}^t,\ldots,  {\bar g}_{m,K}^t \right) $ in \eqref{eq:likelihood}, we need the marginalization with respective to ${\bar g}_{m,1}^t, \ldots, {\bar g}_{m,k-1}^t,{\bar g}_{m,k+1}^t,\ldots, {\bar g}_{m,K}^t $, where ${\bar g}_{k,m}^t \sim \mathcal{N} \left( 0, \nu_k^t \right)$. Then, the marginal distribution $P \! \left( \left. y_m^t \right| { g}_{m,k}^t \right)$, we compute the numerator in \eqref{MMSE element} as 
\begin{align*}
  & \int_{-\infty}^{\infty} \bar{g}_{k,m}^tP \left( \left. y_m^t \right| \bar{g}_{k,m}^t \right) P \left( \bar{g}_{k,m}^t \right)\, d\bar{g}_{k,m}^t \\
    &= \frac{1}{2^{K-1}\sqrt{2\pi\sigma^2}}\sqrt{\frac{2}{\pi}} \nu_k^t\left[ \sum_{\mathbf{b} \in \mathcal{B}_{K,k}} \exp \left\{ - \frac{\left( y_m^t - (\mathbf{h}^t)^{\top} \mathbf{b} \right)^2}{2\sigma_z^2} \right\} \right. \\
    &~~~~~~~~~~~~~~~~~~~~~~~~~ \left. - \sum_{\mathbf{b} \in   \mathcal{B}_{K,k}^c} \exp \left\{ - \frac{\left( y_m^t - (\mathbf{h}^t)^{\top} \mathbf{b} \right)^2}{2\sigma_z^2} \right\} \right],
    \numberthis{\label{numerator}}
\end{align*}
and
\begin{align*}
&  \int_{-\infty}^{\infty} P \left( \left. y_m^t \right| \bar{g}_{k,m}^t \right) P \left( \bar{g}_{k,m}^t \right)\, d\bar{g}_{k,m}^t \\
      &=  \frac{1}{2^{K-1}\sqrt{2\pi\sigma^2}}  \sum_{\mathbf{b} \in \mathcal{B}^K} \exp \left\{ - \frac{\left( y_m^t -(\mathbf{h}^t)^{\top} \mathbf{b} \right)^2}{2\sigma^2} \right\}.
    \numberthis{\label{denominator}}
\end{align*}
The estimated gradient of the $k$th device is derived by substituting \eqref{numerator} and \eqref{denominator} for \eqref{MMSE element}, which arrives at the expression in \eqref{MMSEaggregate}.
\end{proof}

It is instructive to consider special cases for a better understanding of the proposed aggregation function for BayAirComp. 

{\bf Example:} Suppose $K=5$. We first assume that all channel fading coefficients are identical in the magnitude $|h_k^t|=1$ for $k \in [K]$. In this case, as depicted in Fig. \ref{Fig2} (left-side), the aggregation function becomes a uniform soft-step function with maximum and minimum values of $\pm\sqrt{\frac{2}{\pi}}$.  As SNR increases, the soft-step function tends to be sharp. In a heterogeneous fading environment, $|h_1|=5$ and $|h_k|=1$ for $k\in \{2,3,4,5\}$, the proposed aggregation function plays a role of a non-uniform quantizer as illustrated in depicted in Fig. \ref{Fig2} (right-side).  As depicted in Fig. \ref{Fig2}-(a), our proposed BayAirComp clearly differs from the majority-voting based aggregation function. The proposed BayAirComp provides the magnitude information of ${\bar g}_{k,m}^t$ in a quantized manner. Whereas, the majority-voting based aggregation keeps the sign of  ${\bar g}_{k,m}^t$.

{\bf Remark (Implementation):} To implement BayAirComp,  mobile device $k\in [K]$ requires to additionally send  $\mu_k^t\in \mathbb{R}$ and $\nu_k^t\in \mathbb{R}^{+}$ to the server per communication round. As shown in our prior work \cite{lee2020bayesian}, this information can be quantized with $B$-bit scalar quantizer and be transmitted to the server using  orthogonal resources. Since this additional information bits are much smaller than the model size $M\sim 10^{6}$, i.e.,$2B\ll M$, the additional overheads can be negligible.

%Unlike conventional \textit{scalable Bayesian aggregation} function, the new aggregation function increases the number of exponential terms to $2^K$ as the server communicates with multiple mobile devices simultaneously. Fig. \ref{graph_agg} is the graph of sum-gradient estimate calculated by the aggregation function in proposition 1 with some specific channel coefficients. This graph may help to understand the aggregation function intuitively.

\section{Performance Analysis} \label{analysis}
In this section, we provide the convergence analysis of the proposed FL algorithm in this paper. The analysis procedure is carried out in two steps. First, obtain the MSE bound between the true gradient and the gradient estimate by the aggregation function, and then show the gradient of the SGD-based FL algorithm converges to zero. The convergence proof is under the assumption that the global loss function $F \left( \mathbf{w} \right)$ is $L$-Lipschitz smooth and has the least value in $\mathbf{w}^*$. For ease of expression, let define $\mathbf{g}_{\mathsf{true}}^t = \nabla F \left( \mathbf{w}^t \right) = \frac{1}{K} \sum_{k=1}^K \mathbf{g}_k^t$.

\subsection{MSE Bound}

\begin{theorem}
    Let $g_{k,m}^t$ be an IID Gaussian with mean $\mu_k^t$ and variance $\left( \nu_k^t \right)^2$, i.e. $g_{k,m}^t \sim \mathcal{N} \left( \mu_k^t, \left( \nu_k^t \right)^2 \right)$, for $k \in \left[ K \right]$ and $m \in \left[ M \right]$. For the error $\mathbf{e}^t = f_{\mathsf{BayAirComp}} \left( \mathbf{y}^t \right) - \mathbf{g}_{\mathsf{true}}^t$ , the MSE bound $\sigma_{\mathsf{MSE}}^2 \geq \mathbb{E} \left[ \lVert \mathbf{e}^t \rVert_2^2 \right]$ can be expressed as
    \begin{align} \label{MSEbound}
        \sigma_{\mathsf{MSE}}^2 = \frac{M}{K^2} \left( 1 + \frac{2}{\pi} \right) \sum_{k=1}^K \left( \nu_k^t \right)^2.
    \end{align}
\end{theorem}

\begin{proof}
    To reduce the complexity of expressing formulas, we simplify \eqref{MMSEaggregate} as 
    \begin{align}
        f_{\mathsf{BayAirComp}} \left( y_m^t \right) = \frac{1}{K} \sum_{k=1}^K \left[ \mu_k^t + \nu_k^t \sqrt{\frac{2}{\pi}} A_k \! \left( y_m^t \right) \right].
    \end{align}
    Putting the aggregation function into the gradient error definition, we can obtain the formula as
    \begin{align*}
        \frac{1}{K} \sum_{k=1}^K \left[ \mu_k^t + \nu_k^t \sqrt{\frac{2}{\pi}} A_k \! \left( y_m^t \right) \right] &= \frac{1}{K} \sum_{k=1}^K g_{k,m}^t + e_m^t \\
        &= \frac{1}{K} \sum_{k=1}^K \left[ g_{k,m}^t + e_{k,m}^t \right],
        \numberthis{\label{graderror}}
    \end{align*}
    where \eqref{graderror} is for the $m$th component of gradient, and $e_{k,m}^t$ is the error for the $k$th device in $e_m^t$. Then, we compute the MSE bound as follows.
    \begin{align*}
        \mathbb{E} \left[ \lVert \mathbf{e}^t \rVert_2^2 \right] &= \mathbb{E} \left[ \sum_{m=1}^M \lvert e_m^t \rvert^2 \right] \\
        &= \sum_{m=1}^M \mathbb{E} \left[ \lvert e_m^t \rvert^2 \right] \\
        &= \sum_{m=1}^M \mathbb{E} \left[ \left| \frac{1}{K} \sum_{k=1}^K e_{k,m}^t \right|^2 \right] \\
        &\leq \sum_{k=1}^K \mathbb{E} \left[ \frac{1}{K^2} \sum_{k=1}^K \left| e_{k,m}^t \right|^2 \right] \\
        &= \frac{1}{K^2} \sum_{m=1}^M \sum_{k=1}^K \mathbb{E} \left[ \left| e_{k,m}^t \right|^2 \right]. 
        \numberthis{\label{MSEbound1}}
    \end{align*}
    The inequality in \eqref{MSEbound1} is reasonable according to the convexity. Using the assumption that the local gradients $g_{k,m}^t$ and $\bar{g}_{k,m}^t$ is IID Gaussian, it is possible to compute the upper bound of MSE as below. 
    \begin{align*}
        & \mathbb{E} \left[ \lVert \mathbf{e}^t \rVert_2^2 \right] \\
        &\leq \frac{1}{K^2} \sum_{m=1}^M \sum_{k=1}^K \mathbb{E} \left[ \left| g_{k,m}^t - \mu_k^t -\nu_k^t \sqrt{\frac{2}{\pi}} A_k \! \left( y_m^t \right) \right|^2 \right] \\
        &= \frac{1}{K^2} \sum_{m=1}^M \sum_{k=1}^K \mathbb{E} \left[ \left| \bar{g}_{k,m}^t - \nu_k^t \sqrt{\frac{2}{\pi}} A_k \! \left( y_m^t \right) \right|^2 \right] \\
        &= \frac{1}{K^2} \! \sum_{m=1}^M \sum_{k=1}^K \mathbb{E}_{\bar{g}_{k,m}^t} \!\!\! \left[ \left| \bar{g}_{k,m}^t - \nu_k^t \sqrt{\frac{2}{\pi}} \mathbb{E}_{y_m^t} \!\! \left[ A_k \! \left( y_m^t \right) \right] \right|^2 \right] \\
        &\leq \frac{1}{K^2} \sum_{m=1}^M \sum_{k=1}^K \left[ Var\left( g_{k,m}^t \right) + \frac{2}{\pi} \left( \nu_k^t \right)^2 \right] \\
        &= \frac{1}{K^2} \sum_{m=1}^M \sum_{k=1}^K \left( 1 + \frac{2}{\pi} \right) \left( \nu_k^t \right)^2.
        \numberthis{\label{MSEbound2}}
    \end{align*}
    The inequality in \eqref{MSEbound2} is due to the property that
    \begin{align*}
        \mathbb{E}_\mathbf{X} \left[ \left( \mathbf{X} - a \right)^2 \right] &= \sigma_\mathbf{X}^2 + \left( \mu_\mathbf{X} - a \right)^2 \\ 
        &\leq \sigma_\mathbf{X}^2 + \max_{a} \left( \mu_\mathbf{X} - a \right)^2,
        \numberthis{}
    \end{align*}
    and $-1 < A_k \! \left( y_m^t \right) < 1$ regardless of $y_m^t$. Therefore the upper bound of MSE $\sigma_{\mathsf{MSE}}^2$ in \eqref{MSEbound} can be achieved.
\end{proof}

\subsection{Convergence Analysis}

\begin{theorem}
    For the $L$-Lipschitz smooth loss function $F \left( \mathbf{w} \right)$, the proposed FL algorithm with the learning rate $\gamma^t = \frac{\gamma}{t+1}$ for $\gamma > 0$ satisfies 
    \begin{align*} \label{convergencerate}
        & \mathbb{E} \left[ \frac{1}{T} \sum_{t=0}^T \lVert \mathbf{g}_{\mathsf{true}}^t \rVert_2^2 \right] \\
        &\leq \frac{1}{\sqrt{T}} \left[ \frac{F \left( \mathbf{w}^0 \right) - F \left( \mathbf{w}^* \right)}{\gamma \left( 1 - \frac{L\gamma}{2} \right)} + \sigma_{\mathsf{MSE}}^2 \left( 1 + \ln T \right) \frac{\frac{L\gamma}{2}}{1 - \frac{L\gamma}{2}} \right].
        \numberthis{}
    \end{align*}
\end{theorem}

\begin{proof}
    The proposed FL algorithm is based on GD, and the model parameter update formula is given as 
    \begin{align}
        \mathbf{w}^{t+1} = \mathbf{w}^{t} - \gamma^t f_{\mathsf{BayAirComp}} \left( \mathbf{y}^t \right).
    \end{align}
    Since the loss function is $L$-smooth, the convergence formula can be derived as
    \begin{align*}
        & F \left( \mathbf{w}^{t+1} \right) \\
        &\leq F \left( \mathbf{w}^{t} \right) + \left( \mathbf{g}_{\mathsf{true}}^t \right)^T \left( \mathbf{w}^{t+1} - \mathbf{w}^{t} \right) + \frac{L}{2} \lVert \mathbf{w}^{t+1} - \mathbf{w}^{t} \rVert_2^2 \\
        &= F \left( \mathbf{w}^{t} \right) - \left( \mathbf{g}_{\mathsf{true}}^t \right)^T \gamma^t f_{\mathsf{BayAirComp}} \left( \mathbf{y}^t \right) \\
        & \hspace{12em} + \frac{L}{2} \left( \gamma^t \right)^2 \lVert f_{\mathsf{BayAirComp}} \left( \mathbf{y}^t \right) \rVert_2^2 \\
        &= F \left( \mathbf{w}^{t} \right) - \gamma^t \! \left( \mathbf{g}_{\mathsf{true}}^t \right)^T \! \left( \mathbf{g}_{\mathsf{true}}^t + \mathbf{e}^t \right) + \frac{L}{2} \left( \gamma^t \right)^2 \! \lVert \mathbf{g}_{\mathsf{true}}^t + \mathbf{e}^t \rVert_2^2,
        \numberthis{\label{convergence1}}
    \end{align*}
    where the last equation is obtained by the gradient error definition. By taking the expectation in \eqref{convergence1}, we can derive
    \begin{align*} \label{convproblem1}
        & \mathbb{E} \left[ F\left( \mathbf{w}^{t+1} \right) - F\left( \mathbf{w}^{t} \right) \right] \\
        & \leq - \left( \gamma^t - \frac{L}{2} \left( \gamma^t \right)^2 \right) \mathbb{E} \left[ \lVert \mathbf{g}_{\mathsf{true}}^t \rVert_2^2 \right] \\ 
        & \hspace{1.1em} - \left( \gamma^t - L \left( \gamma^t \right)^2 \right) \mathbb{E} \left[ \left( \mathbf{g}_{\mathsf{true}}^t \right)^T \mathbf{e}^t \right] + \frac{L}{2} \left( \gamma^t \right)^2 \mathbb{E} \left[ \lVert \mathbf{e}^t \rVert_2^2 \right].
        \numberthis{}
    \end{align*}
     Usually the learning rate $\gamma^t < 1$, it seems reasonable that $\gamma^t - L \left( \gamma^t \right)^2 > 0$. To continue the convergence analysis, we should find the lower bound of $\mathbb{E} \left[ \left( \mathbf{g}_{\mathsf{true}}^t \right)^T \mathbf{e}^t \right]$.
    
    {\bf Corollary:} If the components of true gradient is Gaussian with mean $\mu_{\mathsf{true}}^t$ and variance $\left( \nu_{\mathsf{true}}^t \right)^2$, i.e. $g_{\mathsf{true}, m}^t \sim \left( \mu_{\mathsf{true}}^t, \left( \nu_{\mathsf{true}}^t \right)^2 \right)$ where $m \in \left[ M \right]$, $\mathbb{E} \left[ \left( \mathbf{g}_{\mathsf{true}}^t \right)^T \mathbf{e}^t \right]$ has a positive value in $\mathsf{SNR} \rightarrow 0$ and $\mathsf{SNR} \rightarrow \infty$.

    %{\bf Corollary:} If the components of true gradient is Gaussian with mean $\mu_{\mathsf{true}}^t$ and variance $\left( \nu_{\mathsf{true}}^t \right)^2$, i.e. $g_{\mathsf{true}, m}^t \sim \left( \mu_{\mathsf{true}}^t, \left( \nu_{\mathsf{true}}^t \right)^2 \right)$ where $m \in \left[ M \right]$, the bound of $\mathbb{E} \left[ \left( \mathbf{g}_{\mathsf{true}}^t \right)^T \mathbf{e}^t \right]$ can be derived as
    
    %\begin{align}
    %     \mathbb{E} \left[ \left( \mathbf{g}_{\mathsf{true}}^t \right)^T \mathbf{e}^t \right] \geq M\left( 1 - \frac{2}{\pi} \sqrt{K} \right) \left( \nu_{\mathsf{true}}^t \right)^2.
    %\end{align}
    
    \begin{proof}
        Firstly, we can derive the mean and variance of $g_{\mathsf{true}, m}^t$ as
        \begin{align*}
            &\mu_{\mathsf{true}}^t = \mathbb{E} \left[ g_{\mathsf{true}, m}^t \right] = \mathbb{E} \left[ \frac{1}{K} \sum_{k=1}^K g_{k,m}^t \right] = \frac{1}{K} \sum_{k=1}^K \mu_k^t, \numberthis{} \\
            &\! \left( \nu_{\mathsf{true}}^t \right)^2 = \mathbb{E} \left[ \left( g_{\mathsf{true}, m}^t - \mu_{\mathsf{true}}^t \right)^2 \right] \\
            & \hspace{3.1em} = \mathbb{E} \left[ \left( \frac{1}{K} \sum_{k=1}^K \left( g_{k,m}^t - \mu_k^t \right) \right)^2 \right] \\
            & \hspace{3.2em} = \mathbb{E} \left[ \frac{1}{K^2} \sum_{k=1}^K \left( g_{k,m}^t - \mu_k^t \right)^2 \right] = \frac{1}{K^2} \sum_{k=1}^K \left( \nu_k^t \right)^2. \numberthis{}
        \end{align*}
        
        We can expressed $\mathbb{E} \left[ \left( \mathbf{g}_{\mathsf{true}}^t \right)^T \mathbf{e}^t \right]$ by summation of each component of the vector as below.
        
        \begin{align*} \label{gebound1}
            &\mathbb{E} \left[ \left( \mathbf{g}_{\mathsf{true}}^t \right)^T \mathbf{e}^t \right] \\
            &= \sum_{m=1}^M \mathbb{E} \left[ g_{\mathsf{true}, m}^t \left( g_{\mathsf{true}, m}^t - f_{\mathsf{BayAirComp}} \left( y_m^t \right) \right) \right] \\
            &= \sum_{m=1}^M \left[ \mathbb{E} \left[ \left( g_{\mathsf{true}, m}^t \right)^2 \right] \right. \\
            & \hspace{5em} \left. - \mathbb{E} \left[ g_{\mathsf{true}, m}^t \times \frac{1}{K} \sum_{k=1}^K \left( \mu_k^t + \sqrt{\frac{2}{\pi}} \nu_k^t A_k \left( y_m^t \right) \right) \right] \right] \\
        \end{align*}
        \begin{align*}
            &= \sum_{m=1}^M \left[ \left( \nu_{\mathsf{true}}^t \right)^2 - \sqrt{\frac{2}{\pi}} \left( \frac{1}{K} \sum_{k=1}^K \nu_k^t \mathbb{E} \left[ g_{\mathsf{true}, m}^t A_k \left( y_m^t \right) \right] \right) \right].
            \numberthis{}
        \end{align*}
        We consider about the exact value of $\mathbb{E} \left[ g_{\mathsf{true}, m}^t A_k \left( y_m^t \right) \right]$ in two $\mathsf{SNR}$ cases to obtain $\mathbb{E} \left[ \left( \mathbf{g}_{\mathsf{true}}^t \right)^T \mathbf{e}^t \right]$: $\mathsf{SNR} \rightarrow 0$, and $\mathsf{SNR} \rightarrow \infty$.
        
        {\bf 1) $\mathsf{\bf{SNR}} \rightarrow \bf{0}$ : }
        In this case, the all exponential terms in \eqref{MMSEaggregate} goes to one as the noise variance $\sigma^2$ goes to infinity, hence $A_k \left( y_m^t \right) = 0$. This derives $\mathbb{E} \left[ g_{\mathsf{true},m}^t A_k \left( y_m^t \right) \right] = 0$, so $\mathbb{E} \left[ \left( \mathbf{g}_{\mathsf{true}}^t \right)^T \mathbf{e}^t \right] = M \left( \nu_{\mathsf{true}}^t \right)^2 > 0$ can be achieved in \eqref{gebound1}.
    
        {\bf 2) $\mathsf{\bf{SNR}} \rightarrow \bf{\infty}$ : }
        The additive noise is assumed to be zero, so the only one exponential term where $\mathbf{b}$ got the whole correct signs of users' gradients is non-zero, and the others are zeros in \eqref{MMSEaggregate}. Therefore we can obtain $A_k \left( y_m^t \right) = \mathsf{sign} \left( \bar{g}_{k,m}^t \right)$. Using this, we can represent the value of $\mathbb{E} \left[ g_{\mathsf{true, m}}^t A_k \left( y_m^t \right) \right]$ as 
    
        \begin{align*}
            & \mathbb{E} \left[ g_{\mathsf{true, m}}^t A_k \left( y_m^t \right) \right] \\
            &= \mathbb{E} \left[ \frac{1}{K} \sum_{\ell = 1}^K g_{\ell, m}^t \mathsf{sign} \left( \bar{g}_{k,m}^t \right) \right] \\
            &= \frac{1}{K} \left[ \mathbb{E} \left[ g_{k,m}^t \mathsf{sign} \left( \bar{g}_{k,m}^t \right) \right] + \sum_{\ell \ne k} \mathbb{E} \left[ g_{\ell,m}^t \mathsf{sign} \left( \bar{g}_{k,m}^t \right) \right] \right]. 
            \numberthis{}
        \end{align*}
        We already know that $g_{k,m}^t$ and $g_{\ell,m}^t \left( \ell \ne k \right)$ are independent, so $\bar{g}_{k,m}^t$ and $g_{\ell,m}^t$ are also independent. By the property that $\mathbb{E} \left[ \mathbf{X} \cdot f\left( \mathbf{Y} \right) \right] = \mathbb{E} \left[ \mathbf{X} \right] \cdot \mathbb{E} \left[ f \left( \mathbf{Y} \right) \right]$ where $\mathbf{X}$ and $\mathbf{Y}$ are independent, $\sum_{\ell \ne k} \mathbb{E} \left[ g_{\ell,m}^t \mathsf{sign} \left( \bar{g}_{k,m}^t \right) \right] = 0$ since $\mathbb{E} \left[ \mathsf{sign} \left( \bar{g}_{k,m}^t \right) \right] = 0$. This helps to obtain the value of $\mathbb{E} \left[ g_{\mathsf{true, m}}^t A_k \left( y_m^t \right) \right]$ as
        
        \begin{align*} \label{Abound}
            \mathbb{E} \left[ g_{\mathsf{true, m}}^t A_k \left( y_m^t \right) \right] &= \frac{1}{K} \mathbb{E} \left[ g_{k,m}^t \mathsf{sign} \left( \bar{g}_{k,m}^t \right) \right] \\
            &= \frac{1}{K} \mathbb{E} \left[ \left( \bar{g}_{k,m}^t + \mu_k^t \right) \mathsf{sign} \left( \bar{g}_{k,m}^t \right) \right] \\
            &= \frac{1}{K} \mathbb{E} \left[ \left| \bar{g}_{k,m}^t \right| \right] = \frac{1}{K} \sqrt{\frac{2}{\pi}} \nu_k^t.
            \numberthis{}
        \end{align*} 
        Put this into \eqref{gebound1}, we can get the exact value of $\mathbb{E} \left[ \left( \mathbf{g}_{\mathsf{true}}^t \right)^T \mathbf{e}^t \right]$ as
        
        \begin{align*}
            & \mathbb{E} \left[ \left( \mathbf{g}_{\mathsf{true}}^t \right)^T \mathbf{e}^t \right] \\
            &= \sum_{m=1}^M \left[ \left( \nu_{\mathsf{true}}^t \right)^2 - \sqrt{\frac{2}{\pi}} \left( \frac{1}{K} \sum_{k=1}^K \nu_k^t \cdot \frac{1}{K} \sqrt{\frac{2}{\pi}} \nu_k^t \right) \right] \\
            &= \sum_{m=1}^M \left[ \left( \nu_{\mathsf{true}}^t \right)^2 - \frac{2}{\pi} \left( \frac{1}{K^2} \sum_{k=1}^K \left( \nu_k^t \right)^2 \right) \right] \\
            &= M \left( 1 - \frac{2}{\pi} \right) \left( \nu_{\mathsf{true}}^t \right)^2 > 0.
            \numberthis{}
        \end{align*}
        
        Consequently, we can summarize that 
        \begin{align}
            \mathbb{E} \left[ \left( \mathbf{g}_{\mathsf{true}}^t \right)^T \mathbf{e}^t \right] = \begin{cases}
                M \left( \nu_{\mathsf{true}}^t \right)^2, & \mathsf{SNR} \rightarrow 0 \\
                M \left( 1 - \frac{2}{\pi} \right) \left( \nu_{\mathsf{true}}^t \right)^2, & \mathsf{SNR} \rightarrow \infty
            \end{cases},
        \end{align}
        and all the values are positive. This completes the proof.

    %    We already know that $-1 < A_k \left( y_m^t \right) < 1$, so we can find the lower bound of \eqref{gebound1} by using $\mathbb{E} \left[ g_{\mathsf{true}, m}^t A_k \left( y_m^t \right) \right] < \mathbb{E} \left[ \left| g_{\mathsf{true}, m}^t \right| \right]$. Therefore we can continue the proof that
    %    \begin{align*}
    %        &\mathbb{E} \left[ \left( \mathbf{g}_{\mathsf{true}}^t \right)^T \mathbf{e}^t \right] \\
    %        &\geq \sum_{m=1}^M \left[ \left( \nu_{\mathsf{true}}^t \right)^2 - \sqrt{\frac{2}{\pi}} \left( \frac{1}{K} \sum_{k=1}^K \nu_k^t \right) \mathbb{E} \left[ \left| g_{\mathsf{true}, m}^t \right| \right] \right] \\
    %        &\geq \sum_{m=1}^M \left[ \left( \nu_{\mathsf{true}}^t \right)^2 - \frac{2}{\pi} \left( \frac{1}{K} \sum_{k=1}^K \nu_k^t \right) \nu_{\mathsf{true}}^t \right] \\
    %        &\geq M \left( 1 - \frac{2}{\pi} \sqrt{K} \right) \left( \nu_{\mathsf{true}}^t \right)^2,
    %        \numberthis{}
    %    \end{align*}
    %    where the last inequality is due to the convexity that 
    %    \begin{align}
    %        \left( \frac{1}{K} \sum_{k=1}^K \nu_k^t \right)^2 \leq \frac{1}{K} \sum_{k=1}^K \left( \nu_k^t \right)^2 = K \left( \nu_{\mathsf{true}}^t \right)^2.
    %    \end{align}
    %    This completes the proof.
    \end{proof}
    
    According to the corollary, \eqref{convproblem1} can be reduced in some particular cases as
    \begin{align*} \label{simpleconv}
        & \mathbb{E} \left[ F \left( \mathbf{w}^{t+1} - F \left( \mathbf{w}^{t} \right) \right) \right] \\
        & \leq - \left( \gamma^t - \frac{L}{2} \left( \gamma^t \right)^2 \right) \mathbb{E} \left[ \lVert \mathbf{g}_{\mathsf{true}}^t \rVert_2^2 \right] + \frac{L}{2} \left( \gamma^t \right)^2 \mathbb{E} \left[ \lVert \mathbf{e}^t \rVert_2^2 \right].
        \numberthis{}
    \end{align*}
    Using the result $\mathbb{E} \left[ \lVert \mathbf{e}^t \rVert_2^2 \right] \leq \sigma_{\mathsf{MSE}}^2$ and the adaptive learning rate $\gamma^t = \frac{\gamma}{t+1} \leq \frac{\gamma}{\sqrt{t+1}}$, we can simplify \eqref{simpleconv} as
    \begin{align*}
        & \mathbb{E} \left[ F \left( \mathbf{w}^{t+1} \right) - F \left( \mathbf{w}^{t} \right) \right] \\
        &\leq - \left( \frac{\gamma}{\sqrt{t+1}} - \frac{L}{2} \frac{\gamma^2}{t+1} \right) \mathbb{E} \left[ \lVert \mathbf{g}_{\mathsf{true}}^t \rVert_2^2 \right] + \frac{L}{2} \frac{\gamma^2}{t+1} \sigma_{\mathsf{MSE}}^2 \\
        &\leq -\frac{\gamma}{\sqrt{t+1}} \mathbb{E} \left[ \lVert \mathbf{g}_{\mathsf{true}}^t \rVert_2^2 \right] \left( 1 - \frac{L\gamma}{2} \right) + \frac{L}{2} \frac{\gamma^2}{t+1} \sigma_{\mathsf{MSE}}^2.
        \numberthis{\label{convergence2}}
    \end{align*}
    If \eqref{convergence2} is summed for all rounds $t \in [T]$, it can be organized as 
    \begin{align*}
        & F \left( \mathbf{w}^0 \right) - F \left( \mathbf{w}^* \right) \\
        &\geq \mathbb{E} \left[ \sum_{t=0}^{T-1} \left( F \left( \mathbf{w}^{t} \right) - F \left( \mathbf{w}^{t+1} \right) \right) \right] \\
        &\geq \sum_{t=0}^{T-1} \left[ \frac{\gamma}{\sqrt{t+1}} \mathbb{E} \left[ \lVert \mathbf{g}_{\mathsf{true}}^t \rVert_2^2 \right] \left( 1 - \frac{L\gamma}{2} \right) - \frac{L}{2} \frac{\gamma^2}{t+1} \sigma_{\mathsf{MSE}}^2 \right] \\
        &\geq \sqrt{T} \gamma \mathbb{E} \left[ \frac{1}{T} \sum_{t=0}^{T-1} \lVert \mathbf{g}_{\mathsf{true}}^t \rVert_2^2 \right] \left( 1 - \frac{L\gamma}{2} \right) - \sum_{t=0}^{T-1} \frac{L}{2} \frac{\gamma^2}{t+1}\sigma_{\mathsf{MSE}}^2 \\
        &\geq \sqrt{T} \gamma \left( 1 - \frac{L\gamma}{2} \right) \mathbb{E} \left[ \frac{1}{T} \sum_{t=0}^{T-1} \lVert \mathbf{g}_{\mathsf{true}}^t \rVert_2^2 \right] \\
        & \hspace{12em} - \left( 1 + \ln{T} \right) \frac{L}{2} \gamma^2 \sigma_{\mathsf{MSE}}^2.
        \numberthis{\label{convergence3}}
    \end{align*}
    The last inequality of \eqref{convergence3} is due to $\sum_{t=0}^{T-1} \frac{1}{t+1} \leq 1 + \ln{T}$. This completes the proof.
\end{proof}

Through the MSE bound and convergence analysis in section \ref{analysis}, we found that the expected value of the gradient norm decreases as the communication round $T$ increases in the order of 
\begin{align} \label{bigOconv}
    \mathcal{O} \left( \frac{c + c' \sigma_{\mathsf{MSE}}^2 \ln{T}}{\sqrt{T}} \right),
\end{align}
for some positive constants $c$ and $c'$. If $\sigma_{\mathsf{MSE}}^2 = 0$, there is no error between the true gradient and gradient estimate and the convergence rate of FL algorithm reduces to $\mathcal{O} \left( \frac{1}{\sqrt{T}} \right)$. Hence the MSE $\sigma_{\mathsf{MSE}}^2$ makes the convergence speed slower. We obtained that the MSE has the constant upper bound, so \eqref{bigOconv} decreases to zero as $T$ goes to infinity because $\lim_{T \rightarrow \infty} \frac{\ln{T}}{\sqrt{T}} = 0$. Finally, we can conclude that the proposed FL algorithm converges to a stationary point as the expected value of gradient goes to zero. Also this analysis can be extended to the algorithm based on SGD using a mini-batch size. 

\section{Simulation Results}
This section provides numerical results to compare the test accuracy of the proposed algorithm and OBDA, a conventional wireless FL scheme. The OBDA system has features of the truncated channel-inversion precoding and majority-voting-based aggregation at the server \cite{zhu2020one}. 

%Hence, with perfect CSIT, the server receives a signal with the summation of true gradients, which are not truncated and additional noise and updates the training model with the sign of the received signal.

{\bf Network model:} We consider a hundred mobile devices, which are uniformly located in a cell with a radius of 1 km. We consider the COST-231 HATA model to take into account path-loss effects between mobile devices and the server and the Rayleigh fading model for small-scale fading effects. 
\begin{figure}[t] 
    \centering 
    \includegraphics[width=1\columnwidth]{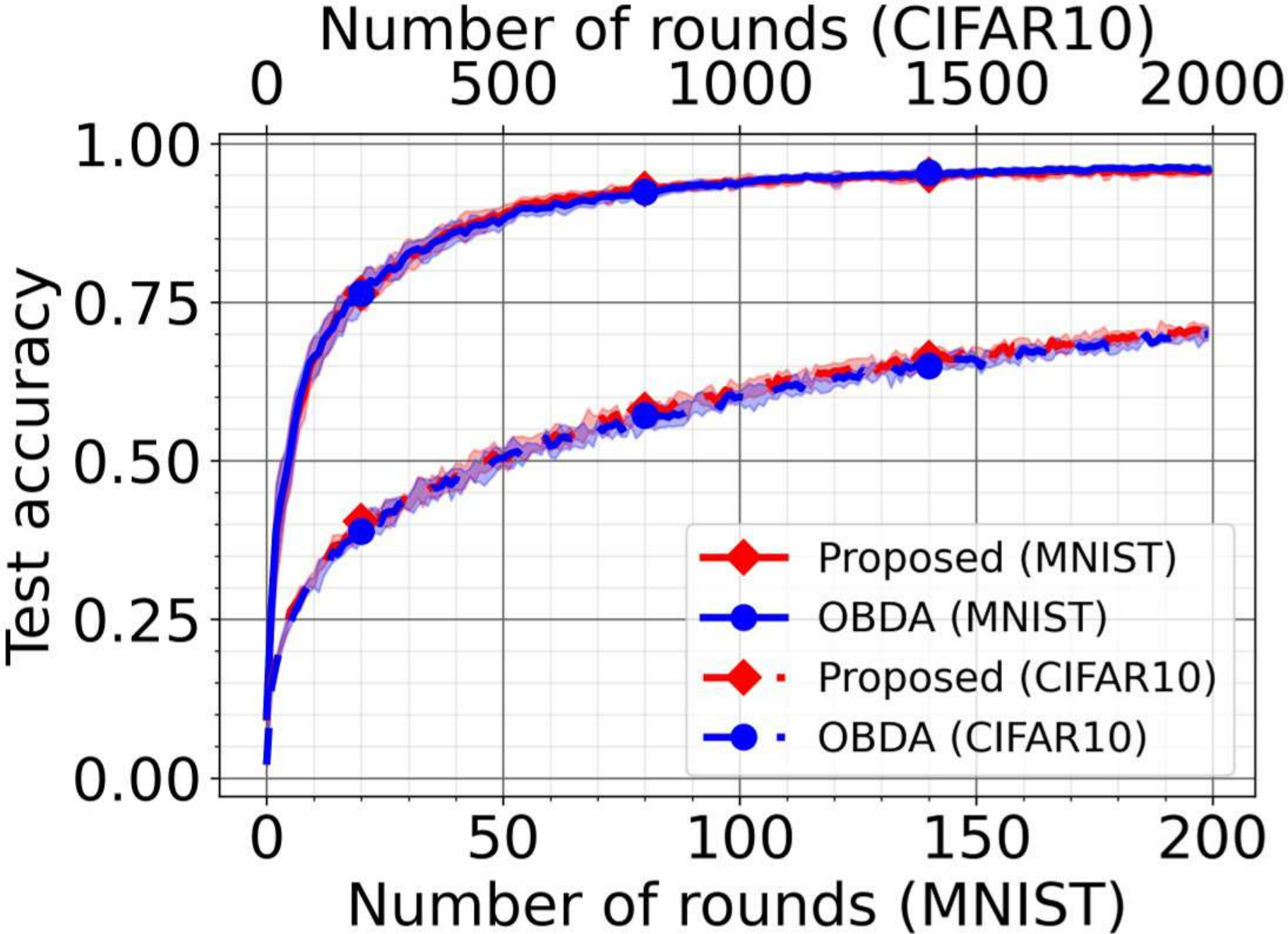}
    \vspace{-0.3cm}\caption{Test accuracy comparison between the proposed system and OBDA system for MNIST and CIFAR10 homogeneous datasets. We use the learning rate of $10^{-3}$ for both algorithms. }\label{homo}
\end{figure}

{\bf Training model:} We consider the task of image classification using MNIST and CIFAR10 datasets. We train a convolutional neural network (CNN)  comprising two $5\times5$ convolutional layers (the first with 32 channels and the second with 64) in \cite{zhu2020one} using  MNIST datasets. We also train ResNet44 model using  CIFAR10 datasets \cite{he2016deep}. To train the model, we assume that the server randomly selects ten mobile users. We also consider two orthogonal time-frequency resources, in which five mobile devices transmit their gradients using a shared time-frequency resource.  For a heterogeneous data assumption, we assign only two distinct types of images to a mobile device. Each mobile device is assumed to compute the local gradient with the same batch size of 32 images. The maximum transmission power is set to be $P=1$.  
% The learning rates for the homogeneous and heterogeneous datasets are chosen to be $10^{-4}$ and $10^{-3}$, respectively. All simulation results are obtained by averaging 30 realizations for each setting. 

%The maximum accuracy for proposed system is about 3.8\% higher than OBDA system. Fig. \ref{CIFAR10_homo} shows the model accuracy comparison between the FL systems which are trained with CIFAR10 datasets. Training is repeated for 5000 rounds, and the accuracy can reached about 70 \%. 

{\bf Effect of sign-alignment precoding:} To see the effect of the proposed sign-alignment precoding, we train the models using the majority-voting based aggregation method as in OBDA, while chaining the precoding strategy from the channel-inversion precoding requiring infinite-resolution CSIT to our sign-alignment precoding using one-bit CSIT.  As can be seen in Fig. \ref{homo},  the both algorithms achieve over 95\% and 70\% test accuracies for MNIST and CIFAR10 datasets, respectively. It is remarkable that our sign-alignment precoding using one-bit CSIT is sufficient for wireless FL systems when the applying signSGD optimizer. This result shows that the sign-information for precoding  degrade the learning performance when applying the channel-inversion precoding.

{\bf Effect of BayAirComp with sign-alignment precoding:} For heterogeneous datasets, we train the models using our BayAirComp aggregator with the sign-alignment precoding.  To improve the convergence speed, the server may harness an accelerated gradient descent algorithm by using a momentum term. To be specific, instead of \eqref{GD}, the server can update the model parameter as
\begin{align} \label{AGD}
    \mathbf{w}^{t+1} = \mathbf{w}^{t} - \gamma^t  \left[\delta f^{\star}({\bf y}^{t-1})+ f^{\star}({\bf y}^t)\right],
\end{align} 
where $\delta\in (0,1)$ is a constant for the moment term with initial value of  $f^{\star}({\bf y}^{0})={\bf 0}$.  As shown in Table \ref{hyperparameter}, we can attain the highest accuracy performance for the proposed FL scheme when hyper-parameters are set to be $\gamma^t=10^{-3}$ and $\delta=0.9$. For OBDA, we set the hyper-parameters to be $\gamma^t=10^{-3}$ and $\delta=0$. Fig. \ref{hetero} shows the test accuracy comparison between OBDA and the proposed algorithm. The proposed algorithm achieve 3.0\% and 3.7\% higher test accuracies than those attained by the OBDA for both MNIST and CIFAR10 dataset, respectively. This result demonstrates that BayAirComp aggregator is beneficial to improve the learning performance for heterogeneous datasets.

%Due to the different gradient sizes computed by the new scheme and OBDA, we simulate with changing the hyperparameter for the scheme using BayAirComp. Table \ref{hyperparameter} contains the converged test accuracy for each learning rate and momentum with MNIST heterogeneous dataset. 

\begin{table}[h]
\centering
\caption{Test accuracies according to different hyper-parameters}
\begin{tabular}{P{.1\columnwidth} P{.1\columnwidth} P{.2\columnwidth}}
    \toprule
    \multicolumn{2}{c}{Hyperparam.} & Test Accuracy \\
    \cmidrule(){1-2}
    $\gamma$ & $\delta$ & \\
    \midrule
    $10^{-2}$ & 0 & 93.70\% \\
    $10^{-2}$ & 0.9 & - \\
    $10^{-3}$ & 0 & 81.81\% \\
    $10^{-3}$ & 0.9 & 94.63\% \\
    $10^{-4}$ & 0 & 71.32\% \\
    $10^{-4}$ & 0.9 & 71.32\% \\
    \bottomrule
\end{tabular}
\label{hyperparameter}
\end{table}

%{\bf Effects of hyper-parameters of the optimizer:}

%For OBDA, the learning rate and momentum are fixed to $10^{-3}$ and $0$ for each. Fig. \ref{hetero} shows the test accuracy comparison between OBDA and the proposed algorithm. The proposed algorihtm attain the 3.0\% and 3.7\% higher test accuracies than those attained by the OBDA for MNIST and CIFAR10 dataset, respectively. This result demonstrates that BayAirComp aggregator is beneficial to improve the learning performance for heterogeneous datasets.

\begin{figure}[t] 
    \centering 
    \includegraphics[width=1\columnwidth]{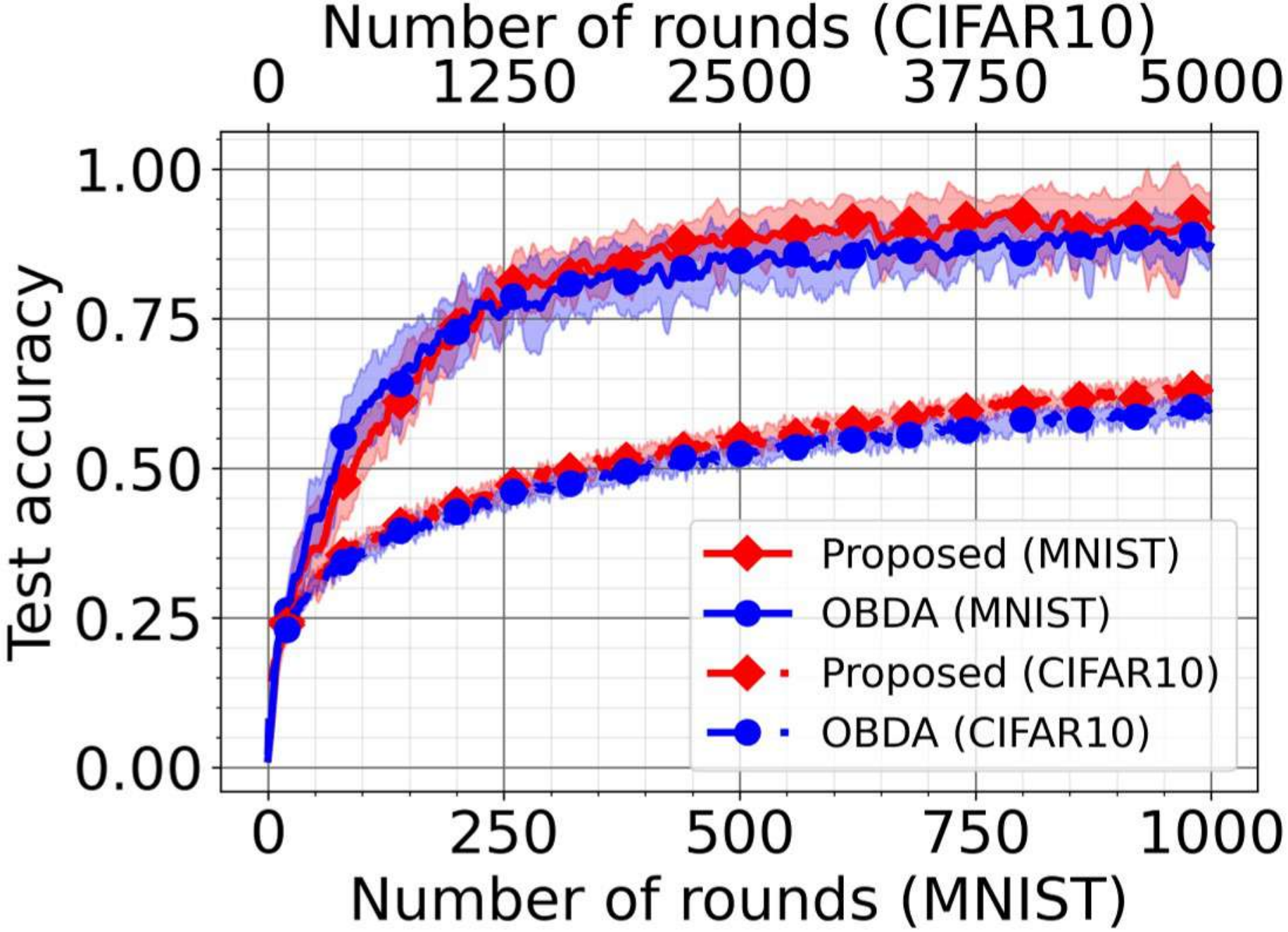}
    \vspace{-0.3cm}\caption{Test accuracy comparison between the proposed and OBDA algorithms for MNIST and CIFAR10 heterogeneous datasets.}\label{hetero}
\end{figure}

\section{Conclusion}
In this work, we studied the problem of wireless federated learning and presented  novel sign-alignment precoding and BayAirComp aggregation method for signSGD. We derived MSE-optimal aggregation function under the IID Gaussian prior of the local gradients when sign-alignment precoding is applied. Our major finding is that one-bit CSIT for precoding  suffices to improve the learning performance compared to the scheme using perfect CSIT. This implies that it is possible to reduce the signaling overheads considerably to implement the wireless FL systems. % From the simulation results, we  numerically demonstrate that with the proposed precoding scheme, one-bit channel information suffices to preserve the accuracy of training models. Moreover, the bayesian-based aggregator achieves improved performance over the conventional counterpart when the data is heterogeneous. 

%In this paper, we propose a novel wireless FL system using AirComp technique. The proposed system uses a precoder that enables efficient learning using only one bit of channel information through limited feedback. In addition, a novel aggregation method based on SBFL is also applied, which enables high-accuracy learning on heterogeneous datasets using quantization functions and prior distribution. By the simulation results, we can confirm that the proposed system exhibits higher test accuracy compared to the conventional FL system.  

%\section{Conclusion 2}
%In this paper, we propose a novel wireless FL system using over-the-air computation technique. The proposed system uses jointly designed precoder and aggregator in orderto reduce the effect of channel faindg in uplink communication. One-bit channel precoding with limited feedback and MMSE estimator with prior distribution of local gradients enable efficient learning for heterogeneous datasets. By the simulation results, we can confirm that the proposed system exhibits higher test accuracy compared to the conventional FL system. One possible direction for the further work would be extend the proposed FL system to a federated multi-task learning system.

\bibliographystyle{IEEEtran}
\bibliography{main}

\begin{thebibliography}{10}
\providecommand{\url}[1]{#1}
\csname url@samestyle\endcsname
\providecommand{\newblock}{\relax}
\providecommand{\bibinfo}[2]{#2}
\providecommand{\BIBentrySTDinterwordspacing}{\spaceskip=0pt\relax}
\providecommand{\BIBentryALTinterwordstretchfactor}{4}
\providecommand{\BIBentryALTinterwordspacing}{\spaceskip=\fontdimen2\font plus
\BIBentryALTinterwordstretchfactor\fontdimen3\font minus
  \fontdimen4\font\relax}
\providecommand{\BIBforeignlanguage}[2]{{%
\expandafter\ifx\csname l@#1\endcsname\relax
\typeout{** WARNING: IEEEtran.bst: No hyphenation pattern has been}%
\typeout{** loaded for the language `#1'. Using the pattern for}%
\typeout{** the default language instead.}%
\else
\language=\csname l@#1\endcsname
\fi
#2}}
\providecommand{\BIBdecl}{\relax}
\BIBdecl

\bibitem{mcmahan2017communication}
B.~McMahan, E.~Moore, D.~Ramage, S.~Hampson, and B.~A. y~Arcas,
  ``Communication-efficient learning of deep networks from decentralized
  data,'' in \emph{Artif. Intell. and Statist.}\hskip 1em plus 0.5em minus
  0.4em\relax PMLR, 2017, pp. 1273--1282.

\bibitem{konevcny2016federated}
J.~Kone{\v{c}}n{\`y}, H.~B. McMahan, F.~X. Yu, P.~Richt{\'a}rik, A.~T. Suresh,
  and D.~Bacon, ``Federated learning: Strategies for improving communication
  efficiency,'' \emph{arXiv preprint arXiv:1610.05492}, 2016.

\bibitem{alistarh2017qsgd}
D.~Alistarh, D.~Grubic, J.~Li, R.~Tomioka, and M.~Vojnovic, ``{QSGD}:
  {Communication-efficient SGD} via gradient quantization and encoding,''
  \emph{Adv. in Neural Inf. Process. Syst.}, vol.~30, pp. 1709--1720, 2017.

\bibitem{reisizadeh2020fedpaq}
A.~Reisizadeh, A.~Mokhtari, H.~Hassani, A.~Jadbabaie, and R.~Pedarsani,
  ``{Fedpaq}: {A} communication-efficient federated learning method with
  periodic averaging and quantization,'' in \emph{Int. Conf. on Artif. Intell.
  and Statist.}\hskip 1em plus 0.5em minus 0.4em\relax PMLR, 2020, pp.
  2021--2031.

\bibitem{chen2016revisiting}
J.~Chen, X.~Pan, R.~Monga, S.~Bengio, and R.~Jozefowicz, ``Revisiting
  distributed synchronous {SGD},'' \emph{arXiv preprint arXiv:1604.00981},
  2016.

\bibitem{li2019convergence}
X.~Li, K.~Huang, W.~Yang, S.~Wang, and Z.~Zhang, ``On the convergence of fedavg
  on non-iid data,'' \emph{arXiv preprint arXiv:1907.02189}, 2019.

\bibitem{zhu2019broadband}
G.~Zhu, Y.~Wang, and K.~Huang, ``Broadband analog aggregation for low-latency
  federated edge learning,'' \emph{IEEE Trans. on Wireless Commun.}, vol.~19,
  no.~1, pp. 491--506, 2019.

\bibitem{zhu2020one}
G.~Zhu, Y.~Du, D.~G{\"u}nd{\"u}z, and K.~Huang, ``One-bit over-the-air
  aggregation for communication-efficient federated edge learning: Design and
  convergence analysis,'' \emph{{IEEE} Trans. on Wireless Commun.}, 2020.

\bibitem{amiri2020machine}
M.~M. Amiri and D.~G{\"u}nd{\"u}z, ``Machine learning at the wireless edge:
  Distributed stochastic gradient descent over-the-air,'' \emph{IEEE Trans. on
  Signal Process.}, vol.~68, pp. 2155--2169, 2020.

\bibitem{amiri2020federated}
------, ``Federated learning over wireless fading channels,'' \emph{IEEE Trans.
  on Wireless Commun.}, vol.~19, no.~5, pp. 3546--3557, 2020.

\bibitem{seif2020wireless}
M.~Seif, R.~Tandon, and M.~Li, ``Wireless federated learning with local
  differential privacy,'' in \emph{2020 IEEE Int. Symp. on Inf. Theory
  (ISIT)}.\hskip 1em plus 0.5em minus 0.4em\relax IEEE, 2020, pp. 2604--2609.

\bibitem{yang2020federated}
K.~Yang, T.~Jiang, Y.~Shi, and Z.~Ding, ``Federated learning via over-the-air
  computation,'' \emph{IEEE Trans. on Wireless Commun.}, vol.~19, no.~3, pp.
  2022--2035, 2020.

\bibitem{shlezinger2020federated}
N.~Shlezinger, M.~Chen, Y.~C. Eldar, H.~V. Poor, and S.~Cui, ``Federated
  learning with quantization constraints,'' in \emph{ICASSP 2020-2020 IEEE Int.
  Conf. on Acoust., Speech and Signal Process. (ICASSP)}.\hskip 1em plus 0.5em
  minus 0.4em\relax IEEE, 2020, pp. 8851--8855.

\bibitem{bernstein2018signsgd}
J.~Bernstein, Y.-X. Wang, K.~Azizzadenesheli, and A.~Anandkumar, ``sign{SGD}:
  {Compressed} optimisation for non-convex problems,'' in \emph{Int. Conf.
  Mach. Learn. (ICML)}.\hskip 1em plus 0.5em minus 0.4em\relax PMLR, 2018, pp.
  560--569.

\bibitem{lee2020bayesian}
S.~Lee, C.~Park, S.-N. Hong, Y.~C. Eldar, and N.~Lee, ``Bayesian federated
  learning over wireless networks,'' \emph{arXiv preprint arXiv:2012.15486},
  2020.

\bibitem{he2016deep}
K.~He, X.~Zhang, S.~Ren, and J.~Sun, ``Deep residual learning for image
  recognition,'' in \emph{Proc. of the IEEE conf. Comput. Vision and Pattern
  Recognit. (CVPR)}, 2016, pp. 770--778.

\end{thebibliography}

\end{document}